\pdfoutput =1

\documentclass[pra,onecolumn,superscriptaddress]{article}

\usepackage[T1]{fontenc}
\usepackage[utf8]{inputenc}
\usepackage{amsmath,amsfonts,amsthm,amssymb}
\usepackage{mathtools}
\usepackage{subeqnarray}
\usepackage{setspace}
\usepackage{graphics,graphicx,color}
\usepackage{url}
\usepackage{enumerate}
\usepackage{float} 
\usepackage{multirow}
\usepackage{hhline}
\usepackage[margin=1.0in]{geometry}
\usepackage{braket}
\usepackage{dsfont} 
\usepackage{authblk}
\usepackage{multirow}
\usepackage{hhline}
\usepackage[makeroom]{cancel}
\usepackage{ifthen}
\usepackage{comment}
\usepackage{breakcites}

\definecolor{ddgreen}{rgb}{.05,.4,.05}
\definecolor{damethyst}{rgb}{0.4, 0.2, 0.6}
\usepackage[colorlinks, linkcolor=damethyst,citecolor=ddgreen]{hyperref}

\newtheorem{theorem}{Theorem}
\newtheorem{cor}{Corollary}
\newtheorem{definition}{Definition}

\newtheorem{lem}{Lemma}

\newtheorem{construction}{Construction}

\usepackage[scr=boondoxo,scrscaled=1.05]{mathalfa}

\newcommand{\Tr}[1]{\mbox{$\mathrm{Tr}\Big[#1\Big]$}}

\newcommand{\beq}{\begin{eqnarray}}
\newcommand{\eeq}{\end{eqnarray}}
\newcommand\abs[1]{\left|#1\right|}

\newcommand{\negl}{\mathsf{negl}}

\newcommand{\sym}{\Pi_{\textnormal{sym}}}

\newcommand{\PRS}{\textnormal{PRS}}
\newcommand{\prf}{\mathsf{PRF}}

\usepackage{tikz}

\definecolor{amethyst}{rgb}{0.6, 0.4, 0.8}

\title{Quantum trapdoor functions from classical one-way functions}
\author{Andrea Coladangelo \thanks{Email: \texttt{coladan@cs.washington.edu}}}
\affil{Paul G. Allen School of Computer Science and Engineering, University of Washington}
\date{}
\begin{document}

\maketitle
\begin{abstract}
We formalize and study the notion of a \emph{quantum trapdoor function}. This is an efficiently computable unitary that takes as input a ``public'' quantum state and a classical string $x$, and outputs a quantum state. This map is such that (i) it is hard to invert, in the sense that it is hard to recover $x$ given the output state (and many copies of the public state), and (ii) there is a classical trapdoor that allows efficient inversion. We show that a quantum trapdoor function can be constructed from any quantum-secure one-way function. A direct consequence of this result is that, assuming just the existence of quantum-secure one-way functions, there exists a public-key encryption scheme with a (pure) \emph{quantum} public key. 
\end{abstract}

\tableofcontents

\section{Introduction}
\emph{One-way} functions and \emph{trapdoor} functions are two of the most well-studied cryptographic primitives. The former are functions that are easy to compute but hard to invert on average. The latter satisfy the same property with the additional feature that there exists some special information, a ``trapdoor'', that enables efficient inversion, but without which inversion remains hard. 

In classical cryptography, there is a marked dividing line between one-way functions and trapdoor functions. One-way functions are equivalent to so-called ``minicrypt'' primitives, e.g. pseudorandom generators, pseudorandom functions, commitments and \emph{private-key} encryption. On the other hand, (injective) trapdoor functions imply \emph{public-key} cryptosystems \cite{yao1982theory, goldwasser1982probabilistic}. This dividing line can be made formal: one can provably show that injective trapdoor functions cannot be constructed from black-box use of one-way functions \cite{impagliazzo1989limits, barak2009merkle}. 

In this work, we consider the question of whether the dividing line between one-way and trapdoor functions also stands in a quantum world, i.e.\ a world where parties have access to quantum computation and communication. Concretely, we ask: 
\begin{center}
        {\em Can we achieve some version of a trapdoor function\\ from a one-way function, in a quantum world?}
\end{center}

\subsection{Our results}
We provide a partial positive answer to the question above. Our first contribution is to formalize the notion of a \emph{quantum trapdoor function} (QTF). This notion was introduced informally in \cite{nikolopoulos2008applications}. A QTF is similar to a classical trapdoor function, but differs in two crucial ways:
\begin{itemize}
    \item The generation procedure outputs a classical trapdoor $tr$ and a \emph{quantum} evaluation key $\ket{eval}$.
    \item Evaluation takes as input a classical string $x$, and the evaluation key $\ket{eval}$, and outputs a \emph{quantum} state $\ket{\psi_x}$.
\end{itemize}
For security, we require that (without the trapdoor) it should be hard to recover $x$ from $\ket{\psi_x}$, \emph{even given access to polynomially many copies of} $\ket{eval}$. On the other hand, given the trapdoor $tr$, inversion is easy. We additionally restrict our attention to the case where (the same) $\ket{eval}$ can be efficiently generated given the trapdoor $tr$. The reason for this is that we think of $\ket{eval}$ as being a ``publicly'' available resource: anyone can request copies of $\ket{eval}$ from the party who has the trapdoor.

Our main result is the following.

\begin{theorem}[informal]
\label{thm: 1}
Quantum trapdoor functions exist, assuming the existence of quantum-secure one-way functions.
\end{theorem}

Our construction is remarkably simple. It uses as a building block a construction of pseudorandom states from one-way functions proposed by Ji, Liu, and Song \cite{ji2018pseudorandom}, and later proven secure by Brakerski and Shmueli \cite{brakerski2019pseudo}. To achieve the trapdoor property we leverage the particular structure of this PRS construction. To prove security, we invoke the PRS property in order to reduce an adversary that inverts our quantum trapdoor function to one that succeeds at a certain information-theoretic ``state-discrimination'' task involving Haar random states.

In classical cryptography, the most common application of trapdoor functions is to construct public-key encryption. Any injective trapdoor function implies a public-key encryption scheme, where the secret key is the trapdoor, and the public key is the (description of the) tradpoor function $f$ itself. In particular, encryption of a message $m \in \{0,1\}$ takes the form $(f(x), hc(x) \oplus m)$, where $x$ is sampled uniformly at random, and $hc(x)$ is a hardcore bit of $x$, e.g.\ the one obtained via Goldreich-Levin \cite{goldreich1989hard}. A quantum trapdoor function yields a public-key encryption scheme in an analogous way, except that the ``public key'' is now the evaluation key $\ket{eval}$, and hence is a quantum state. Thus, the following is a direct consequence of Theorem \ref{thm: 1}. 

\begin{cor}[informal]
\label{cor: 1}
Assuming the existence of quantum-secure one-way functions, there exists a public-key encryption scheme where the public key is a (pure) quantum state.
\end{cor}

The existence of such a scheme was already shown in \cite{morimae2022one}, albeit with a mixed state as a public key\footnote{When the public key is described by a mixed state, either the generator is entangled with the public key, or the public key is essentially a classical distribution. In the latter case, two public keys are ``equal'' as distributions, although in practice they correspond to independently sampled classical strings.}. The construction in \cite{morimae2022one} does not seem to have any implication with regards to constructions of a quantum trapdoor function.

While a quantum state as public key is most likely impractical, and in some sense against the spirit of public-key encryption (where the point is that a party's public-key should be readily available to any party that wishes to send them an encrypted message), we still find the conclusion significant from a conceptual standpoint: the quantum public key is still ``public'' in the sense that the scheme is secure even if an adversary obtains an arbitrary (polynomial) number of copies of the public key. 

Our result suggests that there is some subtlety in defining what it means for quantum information to be ``public'': while in the classical setting having access to a single copy of a string is trivially equivalent to knowing a circuit that generates the string, in the quantum setting, having \emph{access to many copies of a state} is crucially different (in the computational setting) from \emph{knowing a circuit that generates the state}.

We point out that Corollary \ref{cor: 1} implies the following simple \emph{two-message} key-exchange protocol: Alice samples a secret key and a public key, and sends the public key to Bob; then, Bob sends back an encryption of a uniformly sampled key, which Alice is able to decrypt. This protocol is secure in a model where the quantum channel used by Alice and Bob is \emph{authenticated} in the following sense: Alice and Bob can trust the origin of the quantum message that they receive, and that it has not been tampered with, but the adversary obtains a copy of the same quantum message. One can think of this model as capturing a scenario where Alice is broadcasting many copies of the same state (which she can efficiently prepare) to a network of parties which includes Bob and potentially also an adversary. To place this result into context, in the analogous classical setting with authenticated channels, there is a well-known two-message key-exchange protocol based on the hardness of Decision Diffie-Hellman (DDH) \cite{diffie1976new}, but one can provably show that key-exchange, with any number of rounds, \emph{cannot be realized from black-box use of one-way functions} \cite{impagliazzo1989limits}. On the other hand, if one allows for more rounds of communication in the quantum setting, key exchange can of course be realized unconditionally via the BB84 protocol (where \emph{classical} authenticated channels suffice) \cite{bennett1984quantum}. We refer to Section \ref{sec: key exchange} for a more detailed discussion about key exchange, and the need for authenticated channels. 


\subsection{Related and concurrent work}
Our results fit into a broader line of work that aims to understand how the landscape of cryptographic primitives changes in the presence of quantum information. The most well-known result in this direction is that key-exchange (or key-distribution) can be realized unconditionally using quantum communication, via the BB84 protocol \cite{bennett1984quantum}. However, this is somewhat of a standalone result, and does not directly imply that any of the other primitives in ``cryptomania'', e.g.\ trapdoor functions, public-key encryption, oblivious transfer and multiparty-computation, can be realized unconditionally or even from weaker assumptions than what is known classically. In a more recent development, two concurrent works \cite{bartusek2021one, grilo2021oblivious} show that oblivious transfer, and hence multiparty computation, can be realized from one-way functions. One can think of the constructions from these two works as making a more sophisticated use of the original ideas from BB84 in order to instantiate a template for oblivious transfer proposed in \cite{crepeau1988achieving}. 

Our work is the first to draw a connection between trapdoor functions and one-way functions in the quantum setting, and it does so by leveraging a simple novel idea based on the use of pseudorandom states (PRS). 
While in our work we only make use of the structure of a particular construction of PRSs from one-way functions, PRSs have recently received increasing attention more generally. They have been a central object of study in a related line of work that explores the possibility of basing cryptography on the computational hardness of \emph{genuinely quantum} problems (i.e.\ problems with quantum inputs and/or quantum outputs) \cite{ananth2022cryptography, brakerski2022computational, morimae2022quantum}. The existence of PRSs is an example of a computational assumption, involving a quantum problem, that is potentially weaker than the existence of one-way functions. While PRSs can be constructed from one-way functions, the converse is not known to be true: there is in fact evidence to the contrary, namely an oracle with respect to which PRSs exist, but one-way functions do not \cite{kretschmer2021quantum, kretschmer2022quantum}. 

Shortly after the first version of this work appeared, two papers \cite{grilo2023encryption, barooti2023simple} were posted giving simple constructions of public-key encryption schemes with a pure quantum state as a public key (as in our work). The constructions in the two papers \cite{grilo2023encryption, barooti2023simple} are very similar to each other, and they can be seen as improving upon the main idea of \cite{morimae2022one}, which achieved \emph{mixed state} quantum public keys. These constructions have the additional advantage that ciphertexts are \emph{classical} (which is not the case for the construction in our work). However, neither construction seems to imply anything about the existence of quantum trapdoor functions. We consider the latter to be the main contribution of our work.

\subsection{Open questions}
The main open questions left open by our work are: 
\begin{itemize}
\item Can we build a quantum trapdoor function with a \emph{classical} evaluation key from quantum-secure one-way functions? This would imply a public-key encryption scheme with a \emph{classical} public key (and quantum ciphertexts) from one-way functions. Note that the current definition of quantum trapdoor function asks for a classical trapdoor. Does relaxing this requirement to allow for a \emph{quantum} secret key help achieve a scheme with a classical public key? If not, can we prove that this is not possible if one makes black-box use of the one-way function? The current classical black-box separation of public-key encryption and one-way functions does not seem to directly apply as it relies on the ciphertext being classical. 
\item In the converse direction, does the existence of a \emph{quantum} trapdoor function (any of the variants), imply the existence of \emph{classical} one-way function? While in the classical setting trapdoor functions are a special case of one-way functions, it is unclear what the answer to the above question is. It seems plausible that there exists an oracle relative to which quantum trapdoor functions exist, but classical one-way functions do not. An intermediate step in this direction would be to exhibit an oracle relative to which quantum trapdoor functions exists, but classical trapdoor functions do not.
\end{itemize}

\subsection*{Acknowledgements}
The author thanks Tomoyuki Morimae and Takashi Yamakawa for pointing him to the construction and discussion of public-key encryption with quantum public keys in their work \cite{morimae2022one}, and for further discussions about mixed state versus pure state public keys.

\section{Technical Overview}
In this overview, we informally describe the construction of a quantum trapdoor function from one-way functions, and discuss its security. As mentioned, this construction uses as a building block a construction of PRSs from one-way functions due to Ji, Liu, and Song \cite{ji2018pseudorandom}, and later proven secure by Brakerski and Shmueli \cite{brakerski2019pseudo}.

\subsection{Pseudorandom States from one-way functions}
\label{sec: prs}
First, recall what a PRS is. A PRS can be thought of as a \emph{quantum analogue of a pseudorandom generator (PRG)}. A PRS takes as input a classical seed $s \in \{0,1\}^{n}$, where $n$ is a security parameter, and outputs a state $\ket{PRS(s)}$. We ask that the PRS satisfies the following property: it is computationally hard to distinguish between polynomially many copies of $\ket{PRS(s)}$, for a uniformly random $s$, and polynomially many copies of a state sampled from the Haar distribution. More precisely, for any quantum polynomial time $A$, and any $m=poly(n)$, 
$$ \left| \Pr[A(\ket{PRS(s)}^{\otimes m}) = 1: s \gets \{0,1\}^{n}] - \Pr[A(\ket{\psi}^{\otimes m}) = 1  : \ket{\psi} \gets \textnormal{Haar}] \right| = \negl(n) \,.$$

The following is a construction of a PRS from any (quantum-secure) one-way function. The construction is simple (although the proof that it is secure is quite involved). Let $\prf: \mathcal{K} \times \mathcal{X} \leftarrow \{0,1\}$ be a pseudorandom function (PRF) with a one-bit output, where there is an implicit security parameter that we are omitting.
Then the PRS seed is a PRF key $k$, and 
$$ \ket{PRS(k)} = \frac{1}{|\mathcal{X}|}\sum_{y \in \mathcal{X}} (-1)^{\prf(k, y)} \ket{y} \,.$$
Invoking the security of the PRF, we deduce that $m$ copies of $\ket{PRS(k)}$ are computationally indistinguishable from $m$ copies of a state that is a uniform superposition with a uniformly random $\pm 1$ phase. Showing that the latter is \emph{statistically} indistinguishable from $m$ copies of a Haar random state is quite involved, and is the main technical contribution of \cite{brakerski2019pseudo}. A simpler proof of this was later given in \cite{ananth2023pseudorandom}.

Note that the state $\ket{PRS(k)}$ can be generated efficiently by (i) preparing a uniform superposition over the elements of $\mathcal{X}$, (ii) initializing a second register in the state $\ket{-}$, and (iii) computing $\prf$ in superposition, treating the second register as the output register. 

\subsection{Quantum trapdoor functions}
As informally described earlier, a quantum trapdoor function (QTF) consists of the following quantum polynomial time algorithms:
\begin{itemize}
\item[(i)] A generation procedure that outputs a classical trapdoor $tr$ and a \emph{quantum} evaluation key $\ket{eval}$. Additionally we ask that there be an efficient algorithm to generate (the same) $\ket{eval}$ given $tr$.
\item[(ii)] An evaluation procedure that takes as input a string $x$, and $\ket{eval}$, and outputs a \emph{quantum} state $\ket{\psi_x}$. 
\item[(iii)] An inversion procedure that takes as input the trapdoor $tr$ and a state $\ket{\psi_x}$ and returns $x$.
\end{itemize}
For security, we require that, without knowing $tr$, it is hard to recover $x$ given $\ket{\psi_x}$ and an arbitrary (polynomial) number of copies of $\ket{eval}$.  

\paragraph{Construction}
The simple idea behind our construction of a QTF is the following. Consider the PRS construction described in Subsection \ref{sec: prs}. We will take the trapdoor of our QTF to be a PRF key $k$ (uniformly sampled), and the evaluation key to be $\ket{eval} = \ket{PRS(k)}$. For concreteness, for security parameter $n$, take the domain $\mathcal{X}$ of $\prf$ to be $\{0,1\}^n$, so that $\ket{\PRS(k)}$ is an $n$-qubit state. Then, the evaluation of the QTF on input $x$, and evaluation key $\ket{PRS(k)}$, is
$$ \ket{\psi_x} = Z^x \ket{PRS(k)} = \frac{1}{\sqrt{2^{n}}}\sum_{y \in \{0,1\}^n} (-1)^{x\cdot y +\prf(k, y)} \ket{y}\,.$$
Here $Z^x$ denotes the $n$-qubit operator that applies Pauli $Z$ to the $i$-th qubit or not based on the value of the $i$-th bit $x_i$.

To invert, i.e.\ to recover $x$ given $\ket{\psi_x}$ and the trapdoor $k$, simply ``undo'' the PRF phase, namely apply the unitary $G_{\prf}(k): \ket{y} \mapsto (-1)^{\prf(k,y)} \ket{y}$. The crucial observation is that $G_{\prf}(k)$ commutes with $Z^x$. Thus, the PRF phase can be undone ``out of order'', and 
$$ G_{\prf}(k) \ket{\psi_x} = \frac{1}{\sqrt{2^{n}}}\sum_{y \in \{0,1\}^n} (-1)^{x\cdot y} \ket{y} =  H^{\otimes n} \ket{x} \,.$$
Finally, applying $H^{\otimes n}$ returns $x$.

\paragraph{Security}
Recall that we wish to establish that, without the trapdoor, it is hard to recover $x$ given $\ket{\psi_x}$ and $m$ copies of the evaluation key $\ket{eval}$, for any $m = poly$. In the case of our construction, this means that it should be hard to recover $x$ from the state 
$$\big(\ket{Z^x} \ket{PRS(k)}\big) \otimes \ket{PRS(k)}^{\otimes m} \,.$$
We can think of this of a ``state discrimination'' task: given a (uniformly random) mixed state from the ensemble
\begin{equation} 
\Big\{ \mathbb{E}_k \left(Z^x \ket{PRS(k)}\bra{PRS(k)} Z^x \right) \otimes  (\ket{PRS(k)}\bra{PRS(k)})^{\otimes m} \Big\}_{x \in \{0,1\}^n} \,,
\end{equation}
the task is to guess $x$. We wish to argue that the probability of guessing $x$ is negligible. 

Now, suppose that there is an adversary $A$ that breaks security of our QTF construction, and thus succeeds at the above discrimination task with non-negligible probability. Then, by invoking the security of the PRS, it must be the case the same adversary $A$ also succeeds with non-negligible probability when we replace $\ket{PRS(k)}$ with a Haar random state (otherwise we can construct an adversary that has non-negligible advantage at distinguishing copies of the PRS state from copies of a Haar random state). Thus, $A$ must also succeed with non-negligible probability at the following discrimination task: given a uniformly random mixed state from the ensemble 
\begin{equation} 
\label{eq: state disc 2}
\Big\{ \mathbb{E}_{\ket{\psi} \gets \textnormal{Haar}(2^{n})} \left(Z^x \ket{\psi}\bra{\psi} Z^x \right) \otimes  (\ket{\psi}\bra{\psi})^{\otimes m} \Big\}_{x \in \{0,1\}^n} \,,
\end{equation}
the task is to guess $x$ (where $\ket{\psi} \gets \textnormal{Haar}(2^{n})$ denotes sampling from the Haar distribution on $n$-qubit states).

We get a contradiction by showing that the optimal success probability at the state discrimination task in \eqref{eq: state disc 2} is exponentially small when $m=poly(n)$ (see Lemma \ref{lem: key} for the precise scaling). At first, this might seem slightly surprising because the ``classical analogue'' of this discrimination task is \emph{easy}, and can be solved perfectly. The classical analogue is distinguishing between \emph{distributions over strings} given a single sample. Thinking of the operator $Z^x$ as ``xoring $x$'' (in the Hadamard basis), and realizing that copies of a string are ``for free'' classically, a natural classical analogue of the state discrimination task in \eqref{eq: state disc 2} is: given a single sample drawn from one of the distributions in the following ensemble (where $x$ is picked uniformly at random)
\begin{equation} 
\{ (\psi \oplus x, \psi): \psi \gets \{0,1\}^{n} \}_{x \in \{0,1\}^n} \,,
\end{equation}
guess $x$.
Clearly one can recover $x$ by taking the xor of $\phi $ and $\phi \oplus x$. In contrast, this intuition fails for the state discrimination task in \eqref{eq: state disc 2}: even though one has access to polynomially many copies of $\ket{\psi}$, it is unclear how one can ``compare'' the first copy (to which $Z^x$ was applied) to the others in order to recover $x$. The proof that recovering $x$ can be done at most with exponentially small probability is somewhat involved, and we refer to Section \ref{sec: security} for the details.

\section{Preliminaries}
For $n \in \mathbb{N}$, we denote $[n] = \{1,\ldots,n\}$. For a finite set $X$, we write $x \gets X$ to mean that $x$ is sampled uniformly at random from $S$. 

We think of a quantum algorithm as a uniform family of quantum circuits. If the circuits in the family are polynomial-sized then we say that the algorithm is quantum polynomial time, which we abbreviate as QPT. We use the notation $poly$ to denote an (unspecified) polynomially-bounded function. 

\subsection{Quantum information}
We introduce some facts that we will use later on. For $n \in \mathbb{N}$, denote by $\{\ket{j}: j \in [2^n]\}$ the standard basis of the space of $n$ qubits. Let $Z$ be the Pauli $Z$ operator. For $s \in \{0,1\}^n$, let $Z^s := \bigotimes_{i=1}^n Z^{s_i}$. 
\begin{lem}[Pauli Z twirl]
\label{lem: z twirl}
Let $n,m \in \mathbb{N}$. For any $\ket{\psi} \in (\mathbb{C}^2)^{\otimes n} \otimes (\mathbb{C}^2)^{\otimes m} $,
$$ \mathbb{E}_{s \gets \{0,1\}^n} (Z^s \otimes I) \ket{\psi}\bra{\psi} (Z^s\otimes I) = \sum_{j\in [2^n]} (\ket{j}\bra{j}\otimes I) \ket{\psi}\bra{\psi} (\ket{j}\bra{j} \otimes I) \,.$$
\end{lem}
\begin{proof}
This follows from a straightforward calculation.
\end{proof}

We will also make use of what is referred to as the ``Pretty Good Measurement''. In a state discrimination task, one receives a single copy of a (mixed) state $\rho_x$ from an ensemble $\{\rho_x\}_{x \in \mathcal{X}}$, where $\mathcal{X}$ is some index set. The goal is to correctly guess $x$, where usually $x$ is sampled uniformly at random from $\mathcal{X}$. Let $\sigma := \sum_x \rho_x$. Then, the ``Pretty Good Measurement'' (PGM) is the POVM $\{M_x\}_{x\in \mathcal{X}} \cup \{M_{\perp}\}$ where $M_x = \sigma^{-\frac12} \rho_x \sigma^{-\frac12}$ and $M_{\perp} = \mathds{1}_{\textnormal{Ker}(\sigma)}$ is the projection onto the kernel of $\sigma$.\footnote{Equivalently, to get a POVM without an additional ``$\perp$'' outcome, the POVM element $M_{\perp}$ can be summed to one of the other POVM elements without affecting any of the outcome probabilities}. While in general the PGM may not be the POVM that gives the optimal success probability at the state discrimination task, the following lemma says that the PGM is ``pretty good''.

\begin{lem}[PGM is optimal up to quadratic loss]
\label{lem: pgm}
If $p$ is the optimal success probability at a state discrimination task, then the success probability of the PGM for the same task is at least $p^2$.
\end{lem}

\subsection{Pseudorandom functions}
We recall the notion of \emph{quantum-secure} pseudorandom functions.

\begin{definition}[Quantum-secure pseudorandom function] 
Let $\prf = \{\prf_n\}_{n \in \mathbb{N}}$, $\prf_n: \mathcal{K}_n \times  \mathcal{X}_n \rightarrow  \mathcal{Y}_n$ be an efficiently computable function, where $\mathcal{K}_n$ is referred to as the key space, $\mathcal{X}_n$ as the domain, and $\mathcal{Y}_n$ as the co-domain. We say that $\prf$ is a quantum-secure pseudorandom function if for every (non-uniform) QPT oracle algorithm $A$, there exists a negligible function $\negl$ such that, for all $n \in \mathbb{N}$,
$$ \Big| \Pr_{k\gets \mathcal{K}_n}[A^{\prf(k, \cdot)}(1^n)= 1] - \Pr_{O\gets \mathcal{Y}^{\mathcal{X}}}[A^O(1^n) =1 ]    \Big| \leq \negl(n) \,.$$
\end{definition}
Quantum-secure pseudorandom functions exist, assuming the existence of quantum-secure one-way functions \cite{zhandry2012construct}.

\subsection{Quantum randomness and pseudorandomness}
	\subsubsection{The Haar Measure on Quantum States}
 \label{sec: haar}
	The Haar measure on a $d$-dimensional quantum state is the uniform (continuous) probability distribution over $d$-dimensional quantum states i.e. the uniform distribution over unit vectors in $\mathbb{C}^{d}$. It can be thought of as the quantum analogue of a classical uniform distribution over strings.

For $d,m \in \mathbb{N}$, we denote the density matrix obtained by sampling a state according to the Haar measure on $d$-dimensional states, and outputting $m$ copies of it, as
\begin{equation}
\mathbb{E}_{\ket{\psi} \gets \textnormal{Haar}(d)} \left[(\ket{\psi}\bra{\psi})^{\otimes m}\right] \,.\label{eq: haar average}
\end{equation}
For a Hilbert space $\mathcal{H}$, we sometimes also denote the Haar measure on states in $\mathcal{H}$ as $\textnormal{Haar}(\mathcal{H})$.

We refer to $\textnormal{span}\{\ket{\phi}^{\otimes m}: \ket{\phi} \in \mathbb{C}^d \}$ as the \emph{symmetric subspace} of $(\mathbb{C}^d)^{\otimes m}$. We will later make use of the following characterization of the density matrix in \eqref{eq: haar average}. Let $\sym^{d,m}$ denote the orthogonal projector onto the \emph{symmetric subspace} of $(\mathbb{C}^{d})^{\otimes m}$. We have the following.
\begin{lem}[\cite{harrow2013church}]
\label{lem: haar sym subspace}
Let $d,m \in \mathbb{N}$. Then, 
$$\mathbb{E}_{\ket{\psi} \gets \textnormal{Haar}(d)} \left[(\ket{\psi}\bra{\psi})^{\otimes m}\right] = \binom{d+m-1}{m}^{-1} \sym^{d,m} \,.$$
\end{lem}

We will also use the fact that the symmetric subspace has the following convenient basis. For $d,m \in \mathbb{N}$, define $\mathcal{I}_{d,m} = \{(t_1, .., t_d): t_1,\cdots, t_d \in \mathbb{N}, t_1+\cdots+t_d = m\}$. For a vector $\vec{j} =(j_1,\ldots, j_{m}) \in [d]^{m}$, denote 
by $T(\vec{j})$ its \emph{type}, i.e. $T(\vec{j}) $ is defined to be the vector in $\mathcal{I}_{d,m}$ whose $i$-th entry is the number of times $i$ appears in the string $(j_1,\ldots, j_m)$. For $\vec{t} \in \mathcal{I}_{d,m}$, define
$$ s(\vec{t}) := \binom{m}{\vec{t}}^{-\frac12} \sum_{\vec{j}: T(\vec{j}) = \vec{t}} \ket{j_1,\ldots, j_{m}} \,,$$
where $\binom{m}{\vec{t}} = \frac{m!}{t_1!\ldots t_d!}$. The $s(\vec{t})$ vectors form an orthonormal basis of the symmetric subspace.
\begin{lem}[\cite{harrow2013church}]
\label{lem: basis}
Let $d,m \in \mathbb{N}$. Then, 
$$\textnormal{span}\{\ket{\phi}^{\otimes m}: \ket{\phi} \in \mathbb{C}^d \} = \textnormal{span}\{ \ket{s(\vec{t})}: \vec{t} \in \mathcal{I}_{d,m}\} \,.$$
\end{lem}

\subsubsection{Pseudorandom states}
The notion of pseudorandom quantum states was introduced in \cite{ji2018pseudorandom}. The following is a formal definition.
\begin{definition}[Pseudorandom Quantum State (PRS)]\label{def: prs}
A Pseudorandom Quantum State (PRS) is a pair of $QPT$ algorithms $(\mathsf{GenKey},\mathsf{GenState})$ such that the following holds. There is a family of Hilbert spaces $\{\mathcal{H}_n\}_{n \in \mathbb{N}}$, and a family $\{\mathcal{K}_n\}_{n \in \mathbb{N}}$ of subsets of $\{0,1\}^*$ such that:
	   \begin{itemize}
			\item $\mathsf{GenKey}(1^n) \rightarrow k$: Takes as input a security parameter $n$, and outputs a key $k \in \mathcal{K}_n$.
			\item $\mathsf{GenState}(k) \rightarrow \ket{\PRS(k)}$: Takes as input a key $k \in \mathcal{K}_n$, for some $n$, and outputs a state in $\mathcal{H}_n$. We additionally require that the state on input $k$ be unique, and we denote this as $\ket{\PRS(k)}$.
        \end{itemize}
Moreover, the following holds. For any (non-uniform) QPT quantum algorithm $A$, and any $m = poly$, there exists a negligible function $\negl$ such that, for all $n \in \mathbb{N}$,
			$$
			\abs{\Pr_{k \gets \mathsf{GenKey}(1^n)}[A\big( \ket{\PRS(k)}^{\otimes m(n)} \big) = 1] -
				\Pr_{\ket{\psi} \gets \textnormal{Haar}(\mathcal{H}_n)}[A\big(\ket{\psi}^{\otimes m(n)} \big) = 1]} \leq \negl(n) \,.$$
  \end{definition}

A PRS can be constructed from any quantum-secure one-way function \cite{ji2018pseudorandom}. Here we describe a particularly simple construction of PRSs that was proposed, and conjectured to be secure, in \cite{ji2018pseudorandom}, and later proven secure in \cite{brakerski2019pseudo}. 

\begin{construction}[PRS with binary phase \cite{ji2018pseudorandom, brakerski2019pseudo}]
\label{const: prs binary phase}
Let $\mathsf{PRF} = \{ \mathsf{PRF}_{n}  \}_{n \in \mathbb{N}}$ be a pseudorandom function family, where $\prf_{n}: \mathcal{K}_{n} \times \{0,1\}^{n} \rightarrow \{0,1\}$. Define $(\mathsf{GenKey}, \mathsf{GenState})$ as follows.
\begin{itemize}
\item $\mathsf{GenKey}(1^n) \rightarrow k$: Sample a PRF key $k \gets \mathcal{K}_n$. Output $k$.
\item $\mathsf{GenState}(k) \rightarrow \ket{\PRS(k)} $: On input $k \in \mathcal{K}_n$, output $\ket{\PRS(k)} = \sum_{y \in \{0,1\}^n} (-1)^{\prf(k,y)} \ket{y} $.
\end{itemize}
Note that $\mathsf{GenState}$ can be implemented efficiently by initializing an ancilla qubit in the state $\ket{-}$, and applying the unitary that computes $\prf(k, \cdot)$ into that qubit. 
\end{construction}

\begin{theorem}[\cite{brakerski2019pseudo}]
Construction \ref{const: prs binary phase} is a PRS.
\end{theorem}

\section{Quantum trapdoor functions}
\label{sec: qtfs}
\subsection{Definition}
\begin{definition}[Quantum trapdoor function] 
\label{def: qtf} A quantum trapdoor function (QTF) is a tuple of QPT algorithms $(\mathsf{Sample}, \mathsf{Eval}, \mathsf{Invert})$, where:
\begin{itemize}
\item $\mathsf{GenTR}(1^{n}) \rightarrow tr$: Takes as input a security parameter, and outputs a classical trapdoor $tr$.
\item $\mathsf{GenEV}(tr) \rightarrow \ket{eval}$: Takes as input a trapdoor $tr$, and outputs a state $\ket{eval}$. We require that there is a unique $\ket{eval}$ for each $tr$.
\item $\mathsf{Eval}(\ket{eval}, x) \rightarrow \ket{\phi}$: Takes as input an evaluation key $\ket{eval}$ and a classical string $x$, and outputs a quantum state $\ket{\phi}$.
\item $\mathsf{Invert}(tr, \ket{\phi}) \rightarrow x$: Takes as input a trapdoor $tr$ and a quantum state $\ket{\phi}$ and outputs a classical string $x$.
\end{itemize}
These algorithms should satisfy the following:
\begin{itemize}
\item[(a)] (Hard to invert) For any QPT algorithm $A$, for any $m = poly$, there exists a negligible function $\negl$ such that, for all $n\in \mathbb{N}$,
$$ \Pr\Big[ A\left(\mathsf{Eval}(\ket{eval}, x), \ket{eval}^{\otimes m(n)}\right) = x \,:\, tr \gets\mathsf{GenTR}(1^{n}), \ket{eval} \gets\mathsf{GenEV}(tr), x \leftarrow \{0,1\}^{n}\Big] \leq \negl(n) \,.$$
\item[(b)] (Trapdoor) For all $n \in \mathbb{N}$, $$\Pr\Big[\mathsf{Invert}\big(tr, \mathsf{Eval}(\ket{eval},x)\big)  = x  : tr \gets\mathsf{GenTR}(1^{n}), \ket{eval} \gets\mathsf{GenEV}(tr), x \leftarrow \{0,1\}^{n} \Big] = 1 \,.$$
\end{itemize}
\end{definition}

Note that requirement (b) is implicitly imposing that, for any fixed evaluation key, the induced map $x \rightarrow \ket{\psi_x}$ is ``\emph{injective}'', in the sense that for any honestly generated $\ket{\psi_x}$, there is a \emph{unique} ``inverse'' $x$.

\subsection{Construction}
\begin{theorem}
\label{thm: main}
A quantum trapdoor function exists, assuming the existence of any quantum-secure one-way function.
\end{theorem}

We describe a construction of a quantum trapdoor function from a quantum-secure PRF. Let $\mathsf{PRF} = \{ \mathsf{PRF}_{n}  \}_{n \in \mathbb{N}}$ be a quantum-secure PRF, where $\prf_{n}: \mathcal{K}_{n} \times \{0,1\}^{n} \rightarrow \{0,1\}$. Let $(\mathsf{PRS.GenKey}, \mathsf{PRS.GenState})$ be the PRS from construction \ref{const: prs binary phase}, based on $\mathsf{PRF}$.

For $k \in \mathcal{K}_n$, let $G_{\mathsf{PRF}}(k)$ be the unitary that, for $y \in \{0,1\}^{n}$ acts as $\ket{y} \mapsto (-1)^{\prf_n(k,y)} \ket{y}$. From now on, we will drop the subscript $n$ in $\prf_n(k,y)$ for convenience.

Note that $G_{\mathsf{PRF}}(k)$ can be implemented efficiently by initializing an ancilla qubit in the state $\ket{-}$ and applying the unitary that computes $\prf(k,\cdot)$ into that qubit. 
\begin{construction}
\label{const: main}
Define $(\mathsf{GenTR}, \mathsf{GenEV}, \mathsf{Eval}, \mathsf{Invert})$ as follows:
    \begin{itemize}
    \item $\mathsf{GenTR}(1^{n}) \rightarrow tr$: Sample a key $k \gets \mathsf{PRS.GenKey}(1^n)$. Set $tr =k$.
    \item $\mathsf{GenEV}(tr) \rightarrow \ket{eval}$: Set $$\ket{eval} = \mathsf{PRS.GenState}(k) = \ket{PRS(k)}\,.$$
    \item $\mathsf{Eval}(\ket{eval}, x) \rightarrow \ket{\phi}$: Output $\ket{\phi} = Z^x \ket{eval}$.
    \item $\mathsf{Invert}(tr, \ket{\phi}) \rightarrow x$: Compute $H^{\otimes n} G_{\prf}(tr) \ket{\phi}$, and measure in the standard basis. Output the outcome $x$.
    \end{itemize}
\end{construction}

We show the following.
\begin{theorem}
\label{thm: 4}
    Assuming $\prf$ is a quantum-secure PRF, Construction \ref{const: main} is a quantum trapdoor function.
\end{theorem}

Since a quantum-secure PRF can be built from any quantum-secure one-way function \cite{zhandry2012construct}, Theorem \ref{thm: 4} implies Theorem \ref{thm: main}. The rest of Section \ref{sec: qtfs} is dedicated to proving Theorem \ref{thm: 4}.

The trapdoor property is straightforward to verify. Let $tr$ be in the support of $\mathsf{GenTR}(1^n)$, and let $\ket{eval} = \mathsf{GenEV}(tr)$. Then, $tr = k $ and $\ket{eval} = \ket{PRS(k)}$, for some $k$ in the support of $\mathsf{PRS.GenKey}(1^n)$. Then, for any $x$, we have
$$ \mathsf{Invert}(tr, \mathsf{Eval}(\ket{eval}, x)) = H^{\otimes n} G_{\prf}(k) Z^x \ket{PRS(k)}\,. $$
Notice that, by construction \ref{const: prs binary phase}, $\ket{PRS(k)} = G_\mathsf{PRF}(k) H^{\otimes n} \ket{0}^{\otimes n}$. Then,
$$ \mathsf{Invert}(tr, \mathsf{Eval}(\ket{eval}, x)) = H^{\otimes n}  G_\mathsf{PRF}(k) Z^x  G_\mathsf{PRF}(k) H^{\otimes n} \ket{0}^{\otimes n} \,.  $$
Crucially, notice that $Z^x$ commutes with $G_\mathsf{PRF}(k)$, since they are both diagonal in the standard basis, and moreover notice that $G_\mathsf{PRF}(k)$ is self-inverse. Thus, we have
$$ \mathsf{Invert}(tr, \mathsf{Eval}(\ket{eval}, x)) = H^{\otimes n} Z^x  H^{\otimes n} \ket{0}^{\otimes n} = \ket{x} \,. $$

The crucial part of this calculation is that $Z^x$ commutes with $G_\mathsf{PRF}(k)$, and thus that the $\mathsf{PRF}$ phase can be ``undone'' even \emph{after} the $Z^x$ phase is applied.

\subsection{Security}
\label{sec: security}
In this subsection, we show that Construction \ref{const: main} also satisfies property (a) from Definition \ref{def: qtf}, i.e.\ it is hard to invert without knowing the trapdoor. This will conclude the proof of Theorem \ref{thm: 4}, and thus of Theorem \ref{thm: main}.

Suppose for a contradiction that there exists a QPT algorithm $A$, a function $m = poly$, and a non-negligible function $\textsf{non-negl}$ such that, for all $n \in \mathbb{N}$,
\begin{equation} 
\label{eq: guess} 
\Pr\Big[ A\left(\mathsf{Eval}(\ket{eval}, x), \ket{eval}^{\otimes m(n)}\right) = x \,:\, tr \gets\mathsf{GenTR}(1^{n}), \ket{eval} \gets\mathsf{GenEV}(tr), x \leftarrow \{0,1\}^{n}\Big] \geq \textsf{non-negl}(n) \,.
\end{equation}
Using $A$, we will construct a distinguisher $D$ that breaks the security of the underlying PRS. Let $m$ be as in Equation \eqref{eq: guess}. $D$ is defined as follows, where we omit the security parameter for ease of notation:
\begin{itemize}
\item On input $\ket{\psi}^{\otimes m+1}$ (where $\ket{\psi}$ is either sampled according to the PRS or the Haar random distribution), sample $x \gets \{0,1\}^{n}$.
\item Give $(Z^x \otimes I^{\otimes m} )\ket{\psi}^{\otimes m+1}$ as input to $A$.
\item Let $x'$ be $A$'s output. If $x' = x$, guess that the input came from the PRS distribution. Otherwise, guess that it came from the Haar distribution.
\end{itemize}
Denote by $\Pr[x' = x | \ket{\psi} \gets \textnormal{PRS}]$ the probability that $A$'s guess $x'$ is equal to $x$, when $\ket{\psi}$ is sampled from the PRS distribution, and by $\Pr[x' = x | \ket{\psi} \gets \textnormal{Haar}]$ the analogous probability when $\ket{\psi}$ is sampled from the Haar distribution.
Notice that in the case that $\ket{\psi}$ is sampled from the PRS distribution, the input that $D$ provides to $A$ is distributed exactly as in Equation \eqref{eq: guess}. Thus, the probability that $D$ guesses correctly is
\begin{align}
\Pr[x' = x | \ket{\psi} \gets \textnormal{PRS}] \cdot \frac12 + \Pr[x' \neq x | \ket{\psi} \gets \textnormal{PRS}] \cdot  \frac12 = \textsf{non-negl}(n) \cdot \frac12 + \Pr[x' \neq x| \ket{\psi} \gets \textnormal{Haar}] \cdot \frac12 \,.
\end{align}
In particular, notice that $D$'s distinguishing advantage is non-negligible if $\Pr[x' = x | \ket{\psi} \gets \textnormal{Haar}] $ is negligible. We will show that the latter is the case, which thus implies that the distinguisher $D$ breaks the security of the PRS. 

The problem of guessing $x$, in the case where $\ket{\psi} \gets \textnormal{Haar}$, can be viewed as a ``state discrimination'' problem. Then, the fact that $\Pr[x' = x | \ket{\psi} \gets \textnormal{Haar}] $ is negligible is implied by the following more general lemma, which says that the (information-theoretically) optimal probability of guessing $x$ is exponentially small (when $m = poly(n)$).
\begin{lem}
\label{lem: key}
Let $n, m \in \mathbb{N}$ be such that $2^n>2 (m+1)$. Consider the ensemble of states:
$$\{ \rho_x \}_{x \in \{0,1\}^{n}} =  \Big\{\mathbb{E}_{\ket{\psi}\gets \textnormal{Haar}(2^n)}[(Z^x \otimes I^{\otimes m}) (\ket{\psi}\bra{\psi})^{\otimes m+1}(Z^x \otimes I^{\otimes m})]\Big\}_{x \in \{0,1\}^{n}} $$ 
Then, there is a constant $C>0$, such that, for any POVM $\{M_x\}_{x\in \{0,1\}^{n}}$,
$$ \mathbb{E}_{x\gets \{0,1\}^{n}} \textnormal{Tr}[M_x \rho_x] < C \cdot \left(\frac{m}{2^{n}} + \frac{m^7}{2^{3n}}\right)^{\frac12} \,.$$
\end{lem}
Since in our case $m = poly(n)$, the bound in Lemma \ref{lem: key} is exponentially small in $n$, and hence \\$\Pr[x' = x | \ket{\psi} \gets \textnormal{Haar}]$ is exponentially small in $n$. This implies that the distinguisher $D$ described above breaks the security of the PRS, giving us a contradiction. 

Thus, to conclude our proof of security, i.e.\ that Construction \ref{const: main} satisfies property (a) from Definition \ref{def: qtf}, we are left with proving Lemma \ref{lem: key}.

\begin{proof}[Proof of Lemma \ref{lem: key}]
We consider the ``Pretty Good Measurement'' (PGM) for this discrimination task. We show that the PGM achieves a guessing probability that is $C' \cdot \left(\frac{m}{2^{n}} + \frac{m^6}{2^{2n}} \right)$ for some constant $C'>0$. By Lemma \ref{lem: pgm}, this implies the desired bound of Lemma \ref{lem: key}.

Let $\sigma = \sum_{x \in \{0,1\}^{n}} \rho_x$, and let $\sigma^{-1}$ be its pseudoinverse. The PGM for this discrimination task is $\{M_x\}_{x\in \{0,1\}^{n}} \cup \{M_{\perp}\}$ where $M_x = \sigma^{-\frac12} \rho_x \sigma^{-\frac12}$ and $M_{\perp} = \mathds{1}_{\textnormal{Ker}(\sigma)}$ is the projection onto the kernel of $\sigma$. For ease of notation, let $d = 2^{n}$. For $N \in \mathbb{N}$, denote by $\sym^{d, N}$ the projector onto the symmetric subspace of $(\mathbb{C}^d)^{\otimes N}$.

Now, notice that 
\begin{align}
\sigma &= \sum_x (Z^x\otimes I) \mathbb{E}_{\ket{\psi} \leftarrow \textnormal{Haar}(d)} [\ket{\psi}\bra{\psi}^{\otimes m+1}] (Z^x \otimes I) \\
&= \binom{d+m}{m+1}^{-1} \cdot \sum_{x \in \{0,1\}^n} (Z^x \otimes I) \sym^{d,m+1} (Z^x \otimes I)\\
&= \binom{d+m}{m+1}^{-1} \cdot d \sum_{j \in [d]} (\ket{j}\bra{j} \otimes I) \sym^{d,m+1} (\ket{j}\bra{j} \otimes I) \,,
\end{align}
where the second equality is by Lemma \ref{lem: haar sym subspace}, and the third equality is by Lemma \ref{lem: z twirl}.

Now, define $\tilde{\sigma} = \sum_j (\ket{j}\bra{j} \otimes I)\sym^{d,m+1}( \ket{j}\bra{j} \otimes I) $, and let $\tilde{\sigma}^{-1}$ be its pseudoinverse. Then, we have
$$ \sigma^{-1} = \binom{d+m}{m+1} \cdot \frac{1}{d} \cdot \tilde{\sigma}^{-1} \,.$$

In order to compute $\tilde{\sigma}^{-1}$, we take a closer look at $\tilde{\sigma}$ and write it in terms of an eigenbasis. To do so, we first consider the convenient orthonormal basis for the symmetric subspace described in Subsection \ref{sec: haar}, which we recall here. For $N \in \mathbb{N}$, define $\mathcal{I}_{d,N} = \{(t_1, .., t_d): t_1,\cdots, t_d \in \mathbb{N}, t_1+\cdots+t_d = N\}$. For a vector $\vec{j} =(j_1,\ldots, j_{N}) \in [d]^{N}$, denote 
by $T(\vec{j})$ its \emph{type}, i.e. $T(\vec{j}) $ is defined to be the vector in $\mathcal{I}_{d,N}$ whose $i$-th entry is the number of times $i$ appears in the string $(j_1,\ldots, j_N)$. For $\vec{t} \in \mathcal{I}_{d,N}$, define
$$ s(\vec{t}) := \binom{N}{\vec{t}}^{-\frac12} \sum_{\vec{j}: T(\vec{j}) = \vec{t}} \ket{j_1,\ldots, j_{N}} \,,$$
where $\binom{N}{\vec{t}} = \frac{N!}{t_1!\ldots t_d!}$. The content of Lemma \ref{lem: basis} is that $\{\ket{s(\vec{t})}: \vec{t} \in \mathcal{I}_{d,N}\}$ is an orthonormal basis for the symmetric subspace of $(\mathbb{C}^d)^{\otimes N}$. This implies that $\sym^{d, N} = \sum_{\vec{t} \in \mathcal{I}_{d,N}} \ket{s(\vec{t})}\bra{s(\vec{t})}$.

Now, for $j \in [d]$, $r \in \{0,\ldots, m\}$, let $T^m_{j,r} = \{\vec{t} \in \mathcal{I}_{d,m}: t_j = r \}$. Moreover, for $\vec{t} = (t_1,\ldots, t_d) \in \mathcal{I}_{d,m+1}$ with $t_j \geq 1$, define $\vec{t}_{-j} \in \mathcal{I}_{d, m}$ to be such that its $i$-th entry is 
\begin{equation*}
t'_i := \begin{cases} 
t_i\,\,\, \textnormal{ if } i \neq j\\
t_i-1 \,\,\,\textnormal{ if } i = j
\end{cases}
\end{equation*} 
In other words, $\vec{t}_{-j}$ is identical to $\vec{t}$ except that the $j$-th entry is reduced by $1$ (and hence $\vec{t}_{-j} \in \mathcal{I}_{d,m}$). Then, notice that, for any $\vec{t} \in \mathcal{I}_{d,m+1}$,
\begin{align} 
\ket{s(\vec{t})}  &= \binom{m+1}{\vec{t}}^{-\frac12} \sum_{k \in [d]: t_k \geq 1} \ket{k}\otimes \sum_{\vec{j} \in [d]^m: T(\vec{j}) = \vec{t}_{-k}} \ket{j_1,\ldots, j_{m}} \nonumber \\
&= \binom{m+1}{\vec{t}}^{-\frac12} \sum_{k: t_k \geq 1} \ket{k} \otimes \binom{m}{\vec{t}_{-k}}^{\frac12} \ket{s(\vec{t}_{-k})} \nonumber \\
&= \sum_{k: t_k \geq 1} \sqrt{\frac{t_k}{m+1}} \ket{k}\otimes  \ket{s(\vec{t}_{-k})} \,. \label{eq: 155}
\end{align}
Then, we have
\begin{align}
\tilde{\sigma} &=  \sum_{j \in [d]} (\ket{j}\bra{j} \otimes I) \sym^{d,m+1}( \ket{j}\bra{j} \otimes I) \nonumber \\
&= \sum_{j \in [d]}(\ket{j}\bra{j} \otimes I)\Bigg( \sum_{\vec{t} \in \mathcal{I}_{d,m+1}} \ket{s(\vec{t})}\bra{s(\vec{t})} \Bigg) ( \ket{j}\bra{j} \otimes I) \nonumber \\
&= \sum_{j \in [d]} (\ket{j}\bra{j} \otimes I)\Bigg( \sum_{\vec{t} \in \mathcal{I}_{d,m+1}}  \sum_{k: t_k \geq 1}  \sum_{k': t_{k'} \geq 1} \sqrt{\frac{t_k}{m+1}} \cdot  \sqrt{\frac{t_{k'}}{m+1}} \ket{k}\bra{k'} \otimes \ket{s(\vec{t}_{-k})}\bra{s(\vec{t}_{-k'})}     \Bigg)     ( \ket{j}\bra{j} \otimes I) \nonumber \\
&= \sum_{\vec{t} \in \mathcal{I}_{d,m+1}} \sum_{j: t_j \geq 1} \frac{t_j}{m+1} \ket{j}\bra{j} \otimes \ket{s(\vec{t}_{-j})}\bra{s(\vec{t}_{-j})}  \nonumber\\
&= \sum_{j \in [d]} \ket{j}\bra{j} \otimes \sum_{\vec{t} \in \mathcal{I}_{d,m+1}: t_j \geq 1}  \frac{t_j}{m+1} \ket{s(\vec{t}_{-j})}\bra{s(\vec{t}_{-j})}  \nonumber\\
&= \sum_{j \in [d]} \ket{j}\bra{j} \otimes \sum_{r=0}^m \frac{r+1}{m+1} \sum_{\vec{t} \in T^m_{j,r}} \ket{s(\vec{t})}\bra{s(\vec{t})} \nonumber\\
&= \frac{1}{m+1} \sum_{j \in [d]} \ket{j}\bra{j} \otimes \sum_{\vec{t} \in T^m_{j,0}}  \ket{s(\vec{t})}\bra{s(\vec{t})} + \frac{1}{m+1} \sum_{j \in [d]} \ket{j}\bra{j} \otimes \sum_{r=1}^m (r+1) \cdot \sum_{\vec{t} \in T^m_{j,r}}  \ket{s(\vec{t})}\bra{s(\vec{t})} \,,\label{eq: 120}
\end{align}
where the third equality uses Equation \eqref{eq: 155}.

Equation \eqref{eq: 120} implies that $\{\ket{j} \otimes \ket{s(\vec{t})}: j \in[d], \vec{t} \in \mathcal{I}_{d,m}\}$ is exactly the set of eigenvectors of $\tilde{\sigma}$ with non-zero eigenvalue. Hence, the pseudoinverse $\tilde{\sigma}$ is obtained by taking the inverse of the (non-zero) elements on the diagonal. Thus, we have 
$$\tilde{\sigma}^{-1} = (m+1)\cdot \sum_{j \in [d]} \ket{j}\bra{j} \otimes \sum_{\vec{t} \in T^m_{j,0}}  \ket{s(\vec{t})}\bra{s(\vec{t})} + \frac{m+1}{r+1} \sum_{j \in [d]} \ket{j}\bra{j} \otimes \sum_{r=1}^m  \cdot \sum_{\vec{t} \in T^m_{j,r}}  \ket{s(\vec{t})}\bra{s(\vec{t})} \,.$$
Then,
$$ \tilde{\sigma}^{-\frac12} = \sqrt{m+1}\cdot \sum_{j \in [d]} \ket{j}\bra{j} \otimes \sum_{\vec{t} \in T^m_{j,0}}  \ket{s(\vec{t})}\bra{s(\vec{t})} + \sqrt{\frac{m+1}{r+1}}\sum_{j \in [d]} \ket{j}\bra{j} \otimes \sum_{r=1}^m \cdot \sum_{\vec{t} \in T^m_{j,r}}  \ket{s(\vec{t})}\bra{s(\vec{t})} \,.$$
Notice that, for any $j \in [d]$, we can write $\sym^{d,m} = \sum_{\vec{t} \in \mathcal{I}_{d,m}} \ket{s(\vec{t})}\bra{s(\vec{t})} = \sum_{r = 0}^m \sum_{\vec{t} \in T^m_{j,r}}\ket{s(\vec{t})}\bra{s(\vec{t})}$. Then, we have
\begin{equation}
\label{eq: 8}
\tilde{\sigma}^{-\frac12} = \sqrt{m+1} \cdot \mathds{1} \otimes \sym^{d,m} - \sqrt{m+1} \cdot  \sum_{r=1}^m \left(1-\frac{1}{\sqrt{r+1}}\right) \cdot \sum_{j \in [d]} \ket{j}\bra{j} \otimes \sum_{\vec{t} \in T^m_{j,r}}  \ket{s(\vec{t})}\bra{s(\vec{t})} \,.
\end{equation}

For convenience, for $r \in \{0,\ldots,m\}$, define $Q_r = \sum_{j \in [d]} \ket{j}\bra{j} \otimes \sum_{\vec{t} \in T^m_{j,r}}  \ket{s(\vec{t})}\bra{s(\vec{t})} $. The probability of guessing $x$ when using the PGM is
\begin{align}
\label{eq: 10}
& \mathbb{E}_{x\gets \{0,1\}^{n}} \Tr{\rho_x \sigma^{-\frac12} \rho_x \sigma^{-\frac12}} \nonumber \\
&= \binom{d+m}{m+1} \cdot \frac{1}{d} \cdot \mathbb{E}_{x\gets \{0,1\}^{n}} \Tr{\rho_x \tilde{\sigma}^{-\frac12} \rho_x \tilde{\sigma}^{-\frac12}}  \nonumber \\
&= \binom{d+m}{m+1} \cdot \frac{m+1}{d}  \cdot \mathbb{E}_{x\gets \{0,1\}^{n}} \Tr{\rho_x \left(\mathds{1} \otimes \sym^{d,m} \right)\rho_x \left(\mathds{1} \otimes \sym^{d,m}\right)} \nonumber\\
&+ \binom{d+m}{m+1} \cdot \frac{m+1}{d} \cdot \sum_{r,r'=1}^m \left(1- \frac{1}{\sqrt{r+1}}\right)\left(1-\frac{1}{\sqrt{r'+1}}\right) \mathbb{E}_{x\gets \{0,1\}^{n}} \Tr{\rho_x Q_r \rho_x Q_{r'}} \nonumber\\
&-2\binom{d+m}{m+1} \cdot \frac{m+1}{d} \cdot \sum_{r=1}^m \left(1- \frac{1}{\sqrt{r+1}}\right) \mathbb{E}_x \Tr{\rho_x \left(\mathds{1} \otimes \sym^{d,m} \right) \rho_x Q_r} \,.
\end{align}
where the last equality follows by using Equation \eqref{eq: 8} to replace the two appearances of $\tilde{\sigma}^{-\frac12}$. Now, note that, by Lemma \ref{lem: haar sym subspace}, $\rho_x = \binom{d+m}{m+1}^{-1} (Z^x \otimes \mathds{1}) \sym^{d, m+1} (Z^x \otimes \mathds{1})$. Further, note that $Z^x \otimes \mathds{1}$ commutes with both $\mathds{1} \otimes \sym^{d,m}$ and $Q_r$. Substituting the last expression for $\rho_x$ into Equation \eqref{eq: 10}, and noticing that the $Z^x$'s cancel each other out thanks to the commutation, we obtain:
\begin{align}
&\mathbb{E}_{x\gets \{0,1\}^{n}} \Tr{\rho_x \sigma^{-\frac12} \rho_x \sigma^{-\frac12}} \nonumber\\
&= \binom{d+m}{m+1}^{-1} \cdot \frac{m+1}{d}  \cdot \mathbb{E}_{x\gets \{0,1\}^{n}} \Tr{\sym^{d,m+1}\left(\mathds{1} \otimes \sym^{d,m} \right)\sym^{d,m+1} \left(\mathds{1} \otimes \sym^{d,m}\right)} \nonumber \\
&+ \binom{d+m}{m+1}^{-1} \cdot \frac{m+1}{d} \cdot \sum_{r,r'=1}^m \left(1- \frac{1}{\sqrt{r+1}}\right)\left(1-\frac{1}{\sqrt{r'+1}}\right) \mathbb{E}_{x\gets \{0,1\}^{n}} \Tr{\sym^{d,m+1} Q_r \sym^{d,m+1} Q_{r'}} \nonumber \\
&-2\binom{d+m}{m+1}^{-1} \cdot \frac{m+1}{d} \cdot \sum_{r=1}^m \left(1- \frac{1}{\sqrt{r+1}}\right) \mathbb{E}_x \Tr{\sym^{d,m+1} \left(\mathds{1} \otimes \sym^{d,m} \right) \sym^{d,m+1} Q_r} \,. \nonumber
\end{align}
Clearly both $\sym^{d,m+1} \left(\mathds{1} \otimes \sym^{d,m} \right) \sym^{d,m+1}$ and $Q_r$ are positive semidefinite. Using the fact that $\textnormal{Tr}[AB] \geq 0 $ if $A$ and $B$ are positive semidefinite, we have that $\Tr{\sym^{d,m+1} \left(\mathds{1} \otimes \sym^{d,m} \right) \sym^{d,m+1} Q_r} \geq 0$. Thus, we have
\begin{align}
&\mathbb{E}_{x\gets \{0,1\}^{n}} \Tr{\rho_x \sigma^{-\frac12} \rho_x \sigma^{-\frac12}} \nonumber\\
&\leq \binom{d+m}{m+1}^{-1} \cdot \frac{m+1}{d}  \cdot \Tr{\sym^{d,m+1}\left(\mathds{1} \otimes \sym^{d,m} \right)\sym^{d,m+1} \left(\mathds{1} \otimes \sym^{d,m}\right)} \nonumber \\
&+ \binom{d+m}{m+1}^{-1} \cdot \frac{m+1}{d} \cdot \sum_{r,r'=1}^m \left(1- \frac{1}{\sqrt{r+1}}\right)\left(1-\frac{1}{\sqrt{r'+1}}\right)  \Tr{\sym^{d,m+1} Q_r \sym^{d,m+1} Q_{r'}}  \nonumber \\
&\leq \binom{d+m}{m+1}^{-1} \cdot \frac{m+1}{d}  \cdot \Tr{\sym^{d,m+1}\left(\mathds{1} \otimes \sym^{d,m} \right)\sym^{d,m+1} \left(\mathds{1} \otimes \sym^{d,m}\right)} \label{eq: 11}\\
&+ \binom{d+m}{m+1}^{-1} \cdot \frac{m+1}{d} \cdot \sum_{r,r'=1}^m \Tr{\sym^{d,m+1} Q_r \sym^{d,m+1} Q_{r'}} \,, \label{eq: 12}
\end{align}
where the last line follows from the fact that, for any $r,r'$, $\Tr{\sym^{d,m+1} Q_r \sym^{d,m+1} Q_{r'}} \geq 0$ since both $\sym^{d,m+1} Q_r \sym^{d,m+1} $ and $Q_{r'}$ are positive semidefinite.
Next, we bound each of the terms \eqref{eq: 11} and \eqref{eq: 12} separately. First, notice that $$\Tr{\sym^{d,m+1}\left(\mathds{1} \otimes \sym^{d,m} \right)\sym^{d,m+1} \left(\mathds{1} \otimes \sym^{d,m}\right)} = \Tr{\sym^{d, m+1}}\,.$$
Thus,
\begin{align}
\eqref{eq: 11}
&= \binom{d+m}{m+1}^{-1} \cdot \frac{m+1}{d} \cdot \Tr{\sym^{d,m+1}} \nonumber \\
&= \binom{d+m}{m+1}^{-1} \cdot \frac{m+1}{d} \cdot \binom{d+m}{m+1} \nonumber\\
&= \frac{m+1}{d} \,,
\end{align}
where the second equality is again due to Lemma \ref{lem: haar sym subspace}.

Define, for any $j \in [d]$ and $r \in [m]$, $Q_{j, r} :=  \sum_{\vec{t} \in T^m_{j,r}}  \ket{s(\vec{t})}\bra{s(\vec{t})}$. Then, $Q_r =\sum_{j \in [d]} \ket{j}\bra{j} \otimes Q_{j, r}  $. We have,
\begin{align}
\eqref{eq: 12} &\leq \binom{d+m}{m+1}^{-1} \cdot \frac{m+1}{d} \cdot \sum_{r,r'=1}^m  \Tr{\sym^{d,m+1} Q_r \sym^{d,m+1} Q_{r'}} \nonumber\\
&= \binom{d+m}{m+1}^{-1} \cdot \frac{m+1}{d} \cdot \sum_{r,r'=1}^m \sum_{j, j' \in [d]} \Tr{\sym^{d,m+1} \left(\ket{j}\bra{j} \otimes Q_{j,r} \right) \sym^{d,m+1} \left(\ket{j'}\bra{j'} \otimes Q_{j',r'} \right)} \nonumber\\
&= \binom{d+m}{m+1}^{-1} \cdot \frac{m+1}{d} \cdot \sum_{r,r'=1}^m \sum_{j, j' \in [d]} \Tr{\ket{j'}\bra{j'} \sym^{d,m+1} \ket{j}\bra{j}  \left( \mathds{1} \otimes Q_{j,r} \right) \ket{j}\bra{j} \sym^{d,m+1} \ket{j'}\bra{j'} \left( \mathds{1} \otimes Q_{j',r'} \right)} \,, \label{eq: 13}
\end{align}
where in the last line we are writing $\ket{j}\bra{j}$ and $\ket{j'}\bra{j'}$ as short for $\ket{j}\bra{j} \otimes \mathds{1}$ and $\ket{j'}\bra{j'} \otimes \mathds{1}$ respectively.

For $n \in \mathbb{N}$, $j, j' \in [d]$ and $l, l' \in \{0,\ldots,n\}$, define $T^{n}_{(j,l), (j',l')} := \{\vec{t} \in \mathcal{I}_{d,n}: t_j = l \textnormal{ and } t_{j'} = l'\}$. Then, for any $j,j' \in [d]$, we can write 
\begin{equation}
\label{eq: 14}
\sym^{d, m+1} = \sum_{l,l'=0}^{m+1} \,\, \sum_{\vec{t} \in T^{m+1}_{(j,l), (j',l')} } \ket{s(\vec{t})} \bra{s(\vec{t})} \,.
\end{equation}
Recall that, for $j \in [d]$ and $\vec{t} = (t_1,\ldots, t_d) \in \mathcal{I}_{d,m+1}$ such that $t_j \geq 1$, we defined $\vec{t}_{-j} \in \mathcal{I}_{d, m}$ to be identical to $\vec{t}$ except that the $j$-th entry is reduced by $1$. Then, combining Equations \eqref{eq: 14} and \eqref{eq: 155} implies
\begin{equation}
\label{eq: 16}
\ket{j'}\bra{j'} \sym^{d,m+1} \ket{j}\bra{j} = \ket{j'}\bra{j} \otimes \sum_{l, l' \geq 1}^{m+1}  \sqrt{\frac{l}{m+1}} \cdot \sqrt{\frac{l'}{m+1}}  \sum_{\vec{t} \in T^{m+1}_{(j',l'),(j,l)}} \ket{s(\vec{t}_{-j'})}\bra{s(\vec{t}_{-j})}
\end{equation}
For ease of notation, for $j, j' \in [d]$ and $l,l' \in \{1,\ldots,m+1\}$, let $P_{(j',l'), (j,l)}:= \sum_{\vec{t} \in T^{m+1}_{(j',l'),(j,l)}} \ket{s(\vec{t}_{-j'})}\bra{s(\vec{t}_{-j})}$.
Then, plugging \eqref{eq: 16} into \eqref{eq: 13} (twice) and simplifying, we obtain
\begin{align}
\eqref{eq: 12} &\leq \binom{d+m}{m+1}^{-1} \cdot \frac{1}{d  (m+1)} \cdot \sum_{r,r'=1}^m \sum_{j,j' \in [d]} \sum_{\substack{l,l' \geq 1 \\ l'', l''' \geq 1}}^{m+1} \sqrt{l \cdot l' \cdot l'' \cdot l'''} \Tr{\ket{j'}\bra{j'} \otimes \left(P_{(j',l'), (j,l)} Q_{j,r} P_{(j,l''), (j',l''')}Q_{j',r'}\right) }\nonumber\\
&= \binom{d+m}{m+1}^{-1} \cdot \frac{1}{d  (m+1)} \cdot \sum_{r,r'=1}^m \sum_{j,j' \in [d]} \sum_{\substack{l,l' \geq 1 \\ l'', l''' \geq 1}}^{m+1} \sqrt{l \cdot l' \cdot l'' \cdot l'''} \Tr{P_{(j',l'), (j,l)} Q_{j,r} P_{(j,l''), (j',l''')}Q_{j',r'} }
\end{align}
We can further simplify the expression using $l,l',l'',l''' \leq m+1$ to obtain
\begin{align}
\eqref{eq: 12} &\leq \binom{d+m}{m+1}^{-1} \cdot \frac{m+1}{d} \cdot \sum_{r,r'=1}^m \sum_{j,j' \in [d]} \sum_{\substack{l,l' \geq 1 \\ l'', l''' \geq 1}}^{m+1} \left|\Tr{P_{(j',l'), (j,l)} Q_{j,r} P_{(j,l''), (j',l''')}Q_{j',r'} }\right| \label{eq: 18}
\end{align}
Observe that 
\begin{align}
   Q_{j,r} P_{(j,l''), (j',l''')}Q_{j',r'} = \delta_{r = l''-1} \cdot \delta_{r' = l'''-1} \cdot \sum_{\vec{t} \in T^{m+1}_{(j,l''),(j',l''')}} \ket{s(\vec{t}_{-j})}\bra{s(\vec{t}_{-j'})} \,
\end{align}
Then, with a similar observation, and additionally using the ciclicity of trace, we have
\begin{align}
\Tr{P_{(j',l'), (j,l)} Q_{j,r} P_{(j,l''), (j',l''')}Q_{j',r'}} &= \delta_{r = l -1= l''-1} \cdot \delta_{r' = l'''-1 = l'-1} \cdot \Tr{\sum_{\vec{t} \in T^{m+1}_{(j,l''),(j',l''')}} \ket{s(\vec{t}_{-j'})}\bra{s(\vec{t}_{-j'})}} \nonumber\\
&= \delta_{r = l -1= l''-1} \cdot \delta_{r' = l'''-1 = l'-1} \cdot \Tr{\sum_{\vec{t} \in T^{m}_{(j,l''),(j',l'''-1)}} \ket{s(\vec{t})}\bra{s(\vec{t})}} \nonumber\\
&= \delta_{r = l -1= l''-1} \cdot \delta_{r' = l'''-1 = l'-1} \cdot | T^{m}_{(j,l''),(j',l'''-1)} | \,, \label{eq: 223}
\end{align}
where $ |T^{m}_{(j,l''),(j',l'''-1)}|$ is the cardinality of the set $T^{m}_{(j,l''),(j',l'''-1)}$.

Plugging \eqref{eq: 223} into \eqref{eq: 18}, and simplifying, gives
\begin{align}
\eqref{eq: 12} &\leq \binom{d+m}{m+1}^{-1} \cdot \frac{m+1}{d} \cdot \sum_{r,r'=1}^m \sum_{j,j' \in [d]} |T^{m}_{(j,r+1),(j',r')}| \nonumber\\
&=\binom{d+m}{m+1}^{-1} \cdot \frac{m+1}{d} \cdot \sum_{r,r'=1}^m \left(\sum_{j \in [d]} |T^{m}_{(j,r+1),(j,r')}|  + \sum_{j \neq j' \in [d]} |T^{m}_{(j,r+1),(j',r')}| \right) \,. \label{eq: 22}
\end{align}
Now, let's consider the first summand in the last expression of \eqref{eq: 22}. Notice that $T^{m}_{(j,r+1),(j,r')} = \emptyset$ whenever $r+1 \neq r'$. When $r+1 = r'$, we have $T^{m}_{(j,r+1),(j,r')} = T^m_{j, r'}$. Notice that the size of the latter is just equal to the dimension of the symmetric subspace of $(\mathbb{C}^{d-1})^{\otimes m-r'}$ (since $j$ must appear exactly $r'$ times). Thus, 
$$ |T^{m}_{(j,r+1),(j,r')}| = \delta_{r+1 = r'} \cdot |T^m_{j, r'}| = \delta_{r+1 = r'}  \cdot \binom{d+m-r'-2}{m-r'} \,.$$
Hence, we have 
\begin{align}
   & \binom{d+m}{m+1}^{-1} \cdot \frac{m+1}{d} \cdot \sum_{r,r'=1}^m \sum_{j \in [d]} |T^{m}_{(j,r+1),(j,r')}| \nonumber\\
    &= \binom{d+m}{m+1}^{-1} \cdot \frac{m+1}{d} \cdot \sum_{r'=2}^m \sum_{j \in [d]}  \binom{d+m-r'-2}{m-r'}  \nonumber\\
    &= \sum_{r'=2}^m \binom{d+m}{m+1}^{-1} \cdot (m+1)  \binom{d+m-r'-2}{m-r'} \nonumber \\
    &= \sum_{r'=2}^m \frac{(m+1)!(d-1)!}{(d+m)!}\cdot (m+1) \cdot \frac{(d+m-r'-2)!}{(m-r')!(d-2)!} \nonumber\\
    &= \sum_{r'=2}^m\frac{(m+1) \cdots (m-r'+1) \cdot (d-1)}{(d+m)\cdots(d+m-r'-1)}  \cdot (m+1)\nonumber\\
   &\leq \sum_{r'=2}^m \frac{d-1}{d+m-r'-1} \cdot (m+1) \cdot \left(\frac{m+1}{d}\right)^{r'+1} \nonumber\\ 
   &\leq (m+1) \cdot \sum_{r'=2}^m \left(\frac{m+1}{d}\right)^{r'+1} \nonumber \\
   &= (m+1) \cdot \left(\frac{m+1}{d}\right)^3 \cdot \sum_{r=0}^{m-2} \left(\frac{m+1}{d}\right)^r \nonumber \\
   & = \frac{(m+1)^4}{d^3} \cdot \frac{1-\left(\frac{m+1}{d}\right)^{m-1}}{1-\frac{m+1}{d}} \nonumber \\
   &\leq \frac{(m+1)^4}{d^3} \cdot \frac{1}{1-\frac{m+1}{d}} \quad \quad \quad \textnormal{(as long as $d > m+1$)} \nonumber \\
   &= \frac{(m+1)^4}{d^2(d-m-1)} 
   \leq 2 \frac{(m+1)^4}{d^3} \quad \quad \quad \textnormal{(as long as $d > 2(m+1)$)}
   \,. \label{eq: 23}
\end{align}

Next, consider the second summand in the last expression of \eqref{eq: 22} (corresponding to $j \neq j'$). Notice that $T^{m}_{(j,r+1),(j',r')} = \emptyset$ whenever $r +r' \geq m$. When $r+r' \leq m-1$, the size of $T^{m}_{(j,r+1),(j',r')}$ is equal to the dimension of the symmetric subspace of $(\mathbb{C}^{d-2})^{\otimes m-r-r'-1}$ (since $j$ must appear exactly $r+1$ times, and $j'$ must appear exactly $r'$ times). Thus, we have
\begin{align}
& \binom{d+m}{m+1}^{-1} \cdot \frac{m+1}{d} \cdot \sum_{r,r'=1}^m \sum_{j \neq j' \in [d]} |T^{m}_{(j,r+1),(j',r')}| \nonumber\\ 
 &= \binom{d+m}{m+1}^{-1} \cdot \frac{m+1}{d} \cdot \sum_{r,r'=1}^m \sum_{j \in [d]}  \binom{d+m-r-r'-4}{m-r'-r'-1}  \nonumber\\
 &= \sum_{r,r'=1}^m \binom{d+m}{m+1}^{-1} \cdot d \cdot (m+1) \cdot \binom{d+m-r-r'-4}{m-r'-r'-1} \nonumber \\
 &= \sum_{r,r'=1}^m \frac{(m+1)!(d-1)!}{(d+m)!} \cdot d \cdot (m+1) \cdot \frac{(d+m-r-r'-4)!}{(m-r-r'-1)!(d-3)!} \nonumber \\
 & = \sum_{r,r'=1}^m \frac{d(d-1)(d-2) \cdot (m+1)^2 m\cdots (m-r-r')}{(d+m)\cdots (d+m-r-r'-3)} \nonumber\\
 &\leq \frac{d(d-1)(d-2)}{d+m}\cdot \sum_{r,r'=1}^m \left(\frac{m+1}{d+m-1} \right)^{r+r'+3} \nonumber \\
&\leq  d^2 \cdot m^2 \cdot \left(\frac{m+1}{d+m-1} \right)^{5} \leq  C'' \cdot \frac{m^7}{d^3} \,,\label{eq: 24}
\end{align}
for some constant $C''>0$ independent of $m$ and $d$.
Plugging the bounds \eqref{eq: 23} and \eqref{eq: 24} into \eqref{eq: 22}, we obtain, provided $d>2 (m+1)$:
\begin{equation}
    \eqref{eq: 12} \leq 2\frac{(m+1)^4}{d^3} + C''\cdot \frac{m^7}{d^3}
\end{equation}
Finally, using the bounds we obtained for terms \eqref{eq: 11} and \eqref{eq: 12}, we have, provided $d > 2 (m+1)$:
\begin{align}
\mathbb{E}_{x\gets \{0,1\}^{n}} \Tr{\rho_x \sigma^{-\frac12} \rho_x \sigma^{-\frac12}}  \leq \frac{m+1}{d} + 2\frac{(m+1)^4}{d^3} + C''\cdot \frac{m^7}{d^3} = C' \cdot \left(\frac{m}{d} + \frac{m^7}{d^3}\right) \,,
\end{align}
for some constant $C'>0$ independent of $m$ and $d$.
Since, by Lemma \ref{lem: pgm}, the PGM achieves a success probability that is at most the square of the optimal, this gives the desired bound of Lemma \ref{lem: key} (with $C = \sqrt{C'}$).
\end{proof}

\section{Applications}
We describe two applications of quantum trapdoor functions.
\subsection{Public-key encryption with quantum public key}
\label{sec: pk with qpk}
A public-key encryption scheme with a \emph{quantum} public key has almost the same syntax as a classical public-key encryption scheme, except that the public key is a quantum state $\ket{pk}$, and we additionally require that $\ket{pk}$ be efficiently generatable given the classical secret key. In the (CPA) security game, the adversary has access to an arbitrary polynomial number of copies of $\ket{pk}$. For completeness, we provide a formal definition.

\begin{definition}[Public-key encryption with quantum public key]
\label{def: pk with qpk}
A public-key encryption scheme with quantum public key is a tuple of QPT algorithms $(\mathsf{GenSK}, \mathsf{GenPK}, \mathsf{Enc}, \mathsf{Dec})$ as follows.
\begin{itemize}
\item $\mathsf{GenSK}(1^{n}) \rightarrow sk$: Takes as input a security parameter, and outputs a classical secret key $sk$.
\item $\mathsf{GenPK}(sk) \rightarrow \ket{pk}$: Takes as input the secret key $sk$, and outputs the quantum public key $\ket{pk}$. We additionally require that $\ket{pk}$ be unique given $sk$.
\item $\mathsf{Enc}(\ket{pk}, m) \rightarrow \ket{c}$: Takes as input a copy of the public key $\ket{pk}$ and a message $m\in \{0,1\}^*$, and outputs a quantum ciphertext $\ket{c}$.
\item $\mathsf{Dec}(sk, \ket{c}) \rightarrow m$: Takes as input the secret key $sk$ and a ciphertext $\ket{c}$, and outputs a string $m$.
\end{itemize}
These algorithms should satisfy ``correctness'', i.e.\ for all $m \in \{0,1\}^*$, for all $n$, $$ \Pr[\mathsf{Dec}(sk, \ket{c}) = m :  sk \gets \mathsf{Gen}(1^{n}),  \ket{pk} \gets \mathsf{GenPK}( sk), \ket{c} \gets \mathsf{Enc}(\ket{pk}, m)] = 1 \,.$$
\end{definition}

The definition of CPA security is analogous to the classical one, except that we explicitly give the adversary access to an arbitrary polynomial number of copies of the quantum public key.
\begin{definition}[CPA security with quantum public key]
\label{def: cpa with qpk}
A public-key encryption scheme with quantum public key satisfies CPA security if the following holds. For all QPT algorithms $A_0, A_1$, for all $t = poly$, there exists a negligible function $\negl$ such that, for all $n$,
\begin{align}
\Pr[b = b': &sk \gets \mathsf{GenSK}(1^{n}),  \ket{pk} \gets \mathsf{GenPK}(sk), (m_0, m_1,\ket{aux}) \gets A_0(\ket{pk}^{\otimes t(n)}), \nonumber\\ &b \gets \{0,1\}, \ket{c} \gets \mathsf{Enc}(\ket{pk}, m_b), b' \gets A_1(\ket{aux}, \ket{c})] 
\leq \frac12 + \negl(n)  \,. \nonumber
\end{align}    
\end{definition}

We show that quantum trapdoor functions imply a public-key encryption scheme with quantum public key. The construction is analogous to the construction of a (classical) public-key encryption scheme from injective trapdoor functions\footnote{Recall in particular that, for a quantum trapdoor function, the quantum map $x \rightarrow \ket{\psi_x}$ induced by a fixed evaluation key is ``injective'', in the sense that for each honestly generated $\ket{\psi_x}$ there is a unique inverse $x$.}. We start by describing a construction that supports encryptions of single-bit messages. 

Let $(\mathsf{QTF.GenTR}, \mathsf{QTF.GenEV}, \mathsf{QTF.Eval}, \mathsf{QTF.Invert})$ be a quantum trapdoor function, as in Definition \ref{def: qtf}.

\begin{construction}[Public-key encryption with quantum public key, supporting single-bit messages]
\label{const: 2}
Define $(\mathsf{GenSK}, \mathsf{GenPK}, \mathsf{Enc}, \mathsf{Dec})$ as follows:
    \begin{itemize}
    \item $\mathsf{GenSK}(1^{n}) \rightarrow sk$: Sample a trapdoor $tr \gets \mathsf{QTF.GenTR}(1^{n})$. Set $sk = tr$.
    \item $\mathsf{GenPK}(sk) \rightarrow \ket{pk}$: Let $\ket{eval} = \mathsf{QTF.GenEV}(sk)$. Set $\ket{pk} = \ket{eval}$.
    \item $\mathsf{Enc}\big(\ket{pk}, m \in \{0,1\}\big) \rightarrow \ket{c}$: Sample $r,x \gets \{0,1\}^{n}$. Compute $\ket{\psi_x} \leftarrow \mathsf{QTF.Eval}(\ket{pk}, x)$. Set $$\ket{c} = (\ket{\psi_x}, \,r,\, r \cdot x \oplus m) \,.$$
    \item $\mathsf{Dec}(sk, \ket{c}) \rightarrow m'$: Parse $\ket{c}$ as $\ket{c} = (\ket{\phi}, r, b)$. Compute $x' \gets \mathsf{QTF.Invert}(sk, \ket{\phi})$. Set $m' = r\cdot x' \oplus b$.
    \end{itemize}
\end{construction}

Correctness of this scheme clearly follows from the trapdoor property of the underlying QTF. We show the following.

\begin{theorem}
\label{thm: pk with qpk cpa security}
Construction \ref{const: 2} satisfies CPA security.
\end{theorem}

\begin{proof}
Suppose for a contradiction that there is a QPT adversary $A$, a $t = poly$, and a non-negligible function \textsf{non-negl} such that, for all $n$, 
\begin{align}
\Pr[ A(\ket{c} \otimes \ket{pk}^{\otimes t}) = m  :  sk \gets \mathsf{Gen}(1^{n}), \ket{pk} \gets \mathsf{GenPK}( sk), m \gets \{0,1\}, \, & \ket{c} \gets \mathsf{Enc}(\ket{pk}, m) ] \nonumber\\
&\geq \frac12 + \textsf{non-negl}(n)  \,. \nonumber
\end{align} 
In the case of Construction \ref{const: 2}, this implies that 
\begin{align}
\Pr\Big[ A(\ket{\psi_x}, \,r,\, r \cdot x \oplus m, \ket{eval}^{\otimes t}) = m  :  &\,\, tr \gets \mathsf{QTF.GenTR}(1^{n}), \ket{eval} \gets \mathsf{QTF.GenEV}(tr), m \gets \{0,1\}, \nonumber \\ &r,x \gets \{0,1\}^{n}, \ket{\psi_x} \gets \mathsf{QTF.Eval}(\ket{eval}, x) \Big]
\geq \frac12 + \textsf{non-negl}(n)  \,. \label{eq: 25}
\end{align} 
It is clear that guessing $m$ is equivalent to guessing $r \cdot x$. Thus, the existence of $A$ satisfying \eqref{eq: 25} implies the existence of an algorithm $A'$ such that 
\begin{align}
\Pr\Big[ A'(\ket{\psi_x}, \,r,\, \ket{eval}^{\otimes t}) = r\cdot x:  &\,\, tr \gets \mathsf{QTF.GenTR}(1^{n}), \ket{eval} \gets \mathsf{QTF.GenEV}(tr), \nonumber \\ &r,x \gets \{0,1\}^{n}, \ket{\psi_x} \gets \mathsf{QTF.Eval}(\ket{eval}, x) \Big]
\geq \frac12 + \textsf{non-negl}(n)   \,.
\end{align} 
Invoking the quantum version of Goldreich-Levin \cite{adcock2002quantum}, which crucially works even in the presence of quantum auxiliary information (since the reduction of \cite{adcock2002quantum} only makes a single call to $A$), we obtain an adversary $A''$ that breaks security of the quantum trapdoor function, i.e.\ given as input $(\ket{\psi_x}, \ket{eval}^{\otimes t})$, $A''$ outputs $x$ with non-negligible probability.
\end{proof}

The scheme of Construction \ref{const: 2} supports encryptions of single-bit messages. In the classical setting, it is well-known that, for public-key schemes, CPA security for single-messages is equivalent to CPA security for multiple messages. Thus, a scheme that supports encryptions of single-bit messages can be bootstrapped to a scheme that supports encryptions of messages of arbitrary length by simply encrypting each bit of the message under the same public key. In the case of \emph{quantum} public key, we need to be a bit more careful. This trivial bootstrapping does not work because the encryption algorithm crucially only receives as input a \emph{single} copy of $\ket{pk}$ (this is of course not an issue classically, since the public key can be copied).

This issue can be circumvented by a slight modification of the classical bootstrapping technique that is referred to as ``hybrid encryption'': 
\begin{itemize}
    \item[(i)] We first obtain a scheme that supports encryptions of $n$-bit messages by having the public key consist of $n$ copies of the public key for a single-bit scheme, i.e. $\ket{pk}^{\otimes n}$. 
   Note that this alone does not resolve the issue because the size of the public-key still grows with the size of the message.
    \item[(ii)] We modify encryption as follows. To encrypt a message $m$ (of arbitrary length), use $\ket{pk}^{\otimes n}$ to encrypt a freshly sampled secret key $\tilde{sk}$ of a \emph{private-key} encryption scheme corresponding to security parameter $n$ (we are assuming this secret key is $n$-bits long), which supports encryption of messages of arbitrary length. Use $\tilde{sk}$ to encrypt $m$ (this step is entirely classical). The ciphertext consists of both encryptions.
\end{itemize}

\begin{theorem}
Assuming the existence of quantum-secure one-way functions, there exists a public-key encryption scheme with a quantum public key (as in Definition \ref{def: pk with qpk}).
\end{theorem}

\begin{proof}
 The proof that CPA security is preserved through steps (i) and (ii) follows essentially unchanged from the respective proofs in the classical setting. The reductions in these proofs only make straightline use of the adversaries. Moreover, these reductions are allowed to make use of an arbitrary polynomial number of copies of the public key, since the adversary in the definition of CPA security  (Definition \ref{def: cpa with qpk}) is allowed the same.
\end{proof}

\subsection{Two-message key-exchange}
\label{sec: key exchange}
In this subsection, we informally discuss the implications of the result from the previous subsection on the possibility of realizing \emph{two-message} key-exchange \emph{from one-way functions}.

Any public-key encryption scheme implies a simple two-message key-exchange protocol: (i) Alice samples $sk$ and $pk$, and sends $pk$ to Bob; (ii) Bob samples a uniformly random string $x$ and sends back an encryption of $x$, which Alice is able to decrypt. Alice and Bob set $x$ to be the shared key.

One can construct an analogous two-message key-exchange protocol from a public-key encryption scheme with a \emph{quantum} public key (like the one described in Subsection \ref{sec: pk with qpk}). The only difference is that Alice and Bob's messages are now quantum states. Then, given the results from the previous section, does this mean that quantum communication allows for a two-message key-exchange protocol from one-way functions?

The answer is a bit subtle, and depends on how one models the attacker's access to the quantum communication. In the case of the BB84 quantum key-exchange protocol, only the \emph{classical} communication channels (used by Alice and Bob to perform the steps of ``information reconciliation'' and ``privacy amplification'') are assumed to be authenticated. The quantum channel is \emph{not} assumed to be authenticated, and the attacker is free to tamper arbitrarily with the qubits that are being exchanged (and the protocol itself takes care of the fact that tampering can be detected by Alice and Bob by subsequently exchanging classical messages). Note that the BB84 protocol requires strictly more than two messages. 

If one insists on considering \emph{two-message} key-exchange protocols, then it is easy to see that some form of authenticated quantum channel is \emph{necessary}. If Alice and Bob's messages are not authenticated at all, then the attacker can simply perform a ``man-in-the-middle'' attack: it can impersonate Bob when interacting with Alice, and viceversa, resulting in the attacker sharing one key with Alice and another with Bob (and the two of them not detecting this).

Then, the question becomes: what is the right notion of \emph{authenticated quantum channel}? In the classical setting, an authenticated channel has the following properties:
\begin{itemize}
\item[(a)] Messages sent through it are guaranteed to not be tampered with, and the origin is trusted.
\item[(b)] However, messages sent through the channel are ``in the clear'', i.e.\ they can be read by an attacker.
\end{itemize}

When trying to generalize this definition to the quantum setting, we run into an obstacle: ``reading'' a quantum message is not something that is well-defined. In particular, the only way that (a) can be satisfied is if the attacker acts as the \emph{identity} on the quantum states that are sent through the channel. Then, trivially the attacker does not gain any information about these quantum states.

One non-trivial quantum analogue of the classical authenticated channel described above is the following:
\begin{itemize}
\item[(a')] Same as (a), but for quantum messages.
\item[(b')] The attacker gets an arbitrary (polynomial) number of copies of any quantum message sent through the channel.
\end{itemize}
This model is not ``realistic'' in the sense that copying a quantum message is not possible in general. However, it does reasonably capture a scenario in which the sender is generating many copies of the message, in our case the quantum public key $\ket{pk}$ (which can be generated efficiently given $sk$), and distributing them to parties in a network. Then, the receiver gets one copy of the message, while the attacker may be able to collect many other copies.

Our two-message key-exchange protocol from public-key encryption with a quantum public key is secure as long as Alice's message is sent through a quantum authenticated channel satisfying (a') and (b') (Bob's message need not be sent through an authenticated channel). The security of this key-exchange protocol reduces exactly to the security of the underlying public-key encryption scheme.

\bibliographystyle{alpha}
\bibliography{references}

\begin{thebibliography}{BCKM21}

\bibitem[AC02]{adcock2002quantum}
Mark Adcock and Richard Cleve.
\newblock A quantum goldreich-levin theorem with cryptographic applications.
\newblock In {\em Annual Symposium on Theoretical Aspects of Computer Science},
  pages 323--334. Springer, 2002.

\bibitem[AGQY23]{ananth2023pseudorandom}
Prabhanjan Ananth, Aditya Gulati, Luowen Qian, and Henry Yuen.
\newblock Pseudorandom (function-like) quantum state generators: New
  definitions and applications.
\newblock In {\em Theory of Cryptography: 20th International Conference, TCC
  2022, Chicago, IL, USA, November 7--10, 2022, Proceedings, Part I}, pages
  237--265. Springer, 2023.

\bibitem[AQY22]{ananth2022cryptography}
Prabhanjan Ananth, Luowen Qian, and Henry Yuen.
\newblock Cryptography from pseudorandom quantum states.
\newblock In {\em Advances in Cryptology--CRYPTO 2022: 42nd Annual
  International Cryptology Conference, CRYPTO 2022, Santa Barbara, CA, USA,
  August 15--18, 2022, Proceedings, Part I}, pages 208--236. Springer, 2022.

\bibitem[BB84]{bennett1984quantum}
Charles~H Bennett and Gilles Brassard.
\newblock Quantum cryptography.
\newblock In {\em Proc. IEEE Int. Conf. on Computers, Systems and Signal
  Processing, Bangalore, India}, pages 175--179, 1984.

\bibitem[BCKM21]{bartusek2021one}
James Bartusek, Andrea Coladangelo, Dakshita Khurana, and Fermi Ma.
\newblock One-way functions imply secure computation in a quantum world.
\newblock In {\em Annual International Cryptology Conference}, pages 467--496.
  Springer, 2021.

\bibitem[BCQ22]{brakerski2022computational}
Zvika Brakerski, Ran Canetti, and Luowen Qian.
\newblock On the computational hardness needed for quantum cryptography.
\newblock {\em arXiv preprint arXiv:2209.04101}, 2022.

\bibitem[BMG09]{barak2009merkle}
Boaz Barak and Mohammad Mahmoody-Ghidary.
\newblock Merkle puzzles are optimal-an o (n2)-query attack on any key exchange
  from a random oracle.
\newblock In {\em CRYPTO}, volume 5677, pages 374--390. Springer, 2009.

\bibitem[BMW23]{barooti2023simple}
Khashayar Barooti, Giulio Malavolta, and Michael Walter.
\newblock A simple construction of quantum public-key encryption from
  quantum-secure one-way functions.
\newblock {\em arXiv preprint arXiv:2303.01143}, 2023.

\bibitem[BS19]{brakerski2019pseudo}
Zvika Brakerski and Omri Shmueli.
\newblock (pseudo) random quantum states with binary phase.
\newblock In {\em Theory of Cryptography: 17th International Conference, TCC
  2019, Nuremberg, Germany, December 1--5, 2019, Proceedings, Part I}, pages
  229--250. Springer, 2019.

\bibitem[CK88]{crepeau1988achieving}
Claude Cr{\'e}peau and Joe Kilian.
\newblock Achieving oblivious transfer using weakened security assumptions.
\newblock In {\em [Proceedings 1988] 29th Annual Symposium on Foundations of
  Computer Science}, pages 42--52. IEEE Computer Society, 1988.

\bibitem[DH76]{diffie1976new}
Whitfield Diffie and Martin~E Hellman.
\newblock New directions in cryptography.
\newblock {\em IEEE Transactions On Information Theory}, 22(6), 1976.

\bibitem[GL89]{goldreich1989hard}
Oded Goldreich and Leonid~A Levin.
\newblock A hard-core predicate for all one-way functions.
\newblock In {\em Proceedings of the twenty-first annual ACM symposium on
  Theory of computing}, pages 25--32, 1989.

\bibitem[GLSV21]{grilo2021oblivious}
Alex~B Grilo, Huijia Lin, Fang Song, and Vinod Vaikuntanathan.
\newblock Oblivious transfer is in miniqcrypt.
\newblock In {\em Advances in Cryptology--EUROCRYPT 2021: 40th Annual
  International Conference on the Theory and Applications of Cryptographic
  Techniques, Zagreb, Croatia, October 17--21, 2021, Proceedings, Part II},
  pages 531--561. Springer, 2021.

\bibitem[GM82]{goldwasser1982probabilistic}
Shafi Goldwasser and Silvio Micali.
\newblock Probabilistic encryption \& how to play mental poker keeping secret
  all partial information.
\newblock In {\em Proceedings of the fourteenth annual ACM symposium on Theory
  of computing}, pages 365--377, 1982.

\bibitem[GSV23]{grilo2023encryption}
Alex~B Grilo, Or~Sattath, and Quoc-Huy Vu.
\newblock Encryption with quantum public keys.
\newblock {\em arXiv preprint arXiv:2303.05368}, 2023.

\bibitem[Har13]{harrow2013church}
Aram~W Harrow.
\newblock The church of the symmetric subspace.
\newblock {\em arXiv preprint arXiv:1308.6595}, 2013.

\bibitem[IR89]{impagliazzo1989limits}
Russell Impagliazzo and Steven Rudich.
\newblock Limits on the provable consequences of one-way permutations.
\newblock In {\em Proceedings of the twenty-first annual ACM symposium on
  Theory of computing}, pages 44--61, 1989.

\bibitem[JLS18]{ji2018pseudorandom}
Zhengfeng Ji, Yi-Kai Liu, and Fang Song.
\newblock Pseudorandom quantum states.
\newblock In {\em Advances in Cryptology--CRYPTO 2018: 38th Annual
  International Cryptology Conference, Santa Barbara, CA, USA, August 19--23,
  2018, Proceedings, Part III 38}, pages 126--152. Springer, 2018.

\bibitem[KQST22]{kretschmer2022quantum}
William Kretschmer, Luowen Qian, Makrand Sinha, and Avishay Tal.
\newblock Quantum cryptography in algorithmica.
\newblock {\em arXiv preprint arXiv:2212.00879}, 2022.

\bibitem[Kre21]{kretschmer2021quantum}
William Kretschmer.
\newblock Quantum pseudorandomness and classical complexity.
\newblock In {\em 16th Conference on the Theory of Quantum Computation,
  Communication and Cryptography}, 2021.

\bibitem[MY22a]{morimae2022one}
Tomoyuki Morimae and Takashi Yamakawa.
\newblock One-wayness in quantum cryptography.
\newblock {\em arXiv preprint arXiv:2210.03394}, 2022.

\bibitem[MY22b]{morimae2022quantum}
Tomoyuki Morimae and Takashi Yamakawa.
\newblock Quantum commitments and signatures without one-way functions.
\newblock In {\em Advances in Cryptology--CRYPTO 2022: 42nd Annual
  International Cryptology Conference, CRYPTO 2022, Santa Barbara, CA, USA,
  August 15--18, 2022, Proceedings, Part I}, pages 269--295. Springer, 2022.

\bibitem[Nik08]{nikolopoulos2008applications}
Georgios~M Nikolopoulos.
\newblock Applications of single-qubit rotations in quantum public-key
  cryptography.
\newblock {\em Physical Review A}, 77(3):032348, 2008.

\bibitem[Yao82]{yao1982theory}
Andrew~C Yao.
\newblock Theory and application of trapdoor functions.
\newblock In {\em 23rd Annual Symposium on Foundations of Computer Science
  (SFCS 1982)}, pages 80--91. IEEE, 1982.

\bibitem[Zha12]{zhandry2012construct}
Mark Zhandry.
\newblock How to construct quantum random functions.
\newblock In {\em 2012 IEEE 53rd Annual Symposium on Foundations of Computer
  Science}, pages 679--687. IEEE, 2012.

\end{thebibliography}

\end{document}